\documentclass[journal]{IEEEtran}

\usepackage{amsmath}
\usepackage{amssymb}
\usepackage{enumerate}

\newtheorem{theorem}{Theorem}[section]

\newtheorem{lemma}[theorem]{Lemma}

\newtheorem{example}[theorem]{Example}

\newcommand{\Fq}{\mathbb{F}_q}
\newcommand{\Fp}{\mathbb{F}_p}
\newcommand{\Fpm}{\mathbb{F}_{p^m}}
\newcommand{\Fptwo}{\mathbb{F}_{p^2}}
\newcommand{\Frtwo}{\mathbb{F}_{r^2}}
\newcommand{\Q}{\mathbb{Q}}
\newcommand{\Ba}{\mathbf{a}}
\newcommand{\Bb}{\mathbf{b}}
\newcommand{\Bc}{\mathbf{c}}
\newcommand{\Be}{\mathbf{e}}
\newcommand{\Bx}{\mathbf{x}}
\newcommand{\By}{\mathbf{y}}
\newcommand{\ca}{\mathcal}
\newcommand{\cC}{\mathcal{C}}
\newcommand{\cD}{\mathcal{D}}
\newcommand{\cG}{\mathcal{G}}
\newcommand{\cH}{\mathcal{H}}
\newcommand{\cS}{\mathcal{S}}
\newcommand{\Tr}{\textup{Tr}}

\newcommand{\Gal}{\textup{Gal}}

\title{The Weight Distribution of a Class of Cyclic Codes Related to Hermitian Forms Graphs}
\author{
Shuxing Li,
Sihuang Hu,
Tao Feng,
and
Gennian Ge
\thanks{
The research of S. Hu was supported by the Scholarship Award for Excellent
Doctoral Student granted by Ministry of Education of China.
The research of T. Feng was supported in part by the Fundamental Research Funds for the Central
Universities of China, Zhejiang Provincial Natural Science Foundation (LQ12A01019), National Natural Science
Foundation of China (11201418).
The research of G. Ge was supported by
the National Outstanding Youth Science Foundation of China under Grant
No.~10825103, National Natural Science Foundation of China
under Grant No.~61171198, and Specialized Research Fund for the
Doctoral Program of Higher Education.}
\thanks{The authors are with the Department of Mathematics, Zhejiang University, Hangzhou
310027, Zhejiang, China
(e-mail: lsxlsxlsx1987@sina.com; husihuang@zju.edu.cn;
         tfeng@zju.edu.cn; gnge@zju.edu.cn).}
}

\begin{document}
\maketitle

\begin{abstract}
The determination of weight distribution of cyclic codes involves evaluation
of Gauss sums and exponential sums. Despite of some cases where a neat
expression is available, the computation is generally rather complicated.
In this note, we determine the weight distribution of a class of reducible
cyclic codes whose dual codes may have arbitrarily many zeros.
This goal is achieved by building an unexpected connection between the
corresponding exponential sums and the spectrums of Hermitian forms graphs.
\end{abstract}

\section{Introduction}
For a cyclic code $\cC$ of length $l$ over finite field $\Fp$ with $p$ prime, let $A_i$ be the number of
codewords in $\cC$ of Hamming weight $i$. The weight distribution $\{A_0,A_1,\ldots,A_l\}$
is an important research subject in coding theory.
Let $h(x)$ be the parity check polynomial of $\cC$. We say that $\cC$ is
irreducible (resp. reducible) if $h(x)$ is irreducible (resp. reducible) over $\Fp$.
When $h(x)=h_0(x)h_1(x)\cdots h_{s-1}(x)$ for some distinct irreducible polynomials
$h_i(x)$ over $\Fp$, the code $\cC$ is the dual of a cyclic code with $s$ zeros.


An identity due to McEliece \cite{Mc} shows that weights of irreducible cyclic codes
can be expressed via Gauss sums. So the determination of the weights of irreducible
cyclic codes can be tackled using number theoretic techniques (see \cite{FY,Mc,MR,Vl,Wo}).
However, this problem is extremely difficult in general since the same is
true for the evaluation of Gauss sums.
When an irreducible cyclic code has exactly one nonzero weight, a nice characterization
has been given in \cite{DY,Vega,VW}. Besides, the class of two-weight irreducible
cyclic codes was extensively studied. The necessary and sufficient conditions for an
irreducible cyclic code to have at most two weights were given by Schmidt and White \cite{SW}.
And they conjectured that all irreducible two-weight cyclic codes consist of
two infinite families and eleven sporadic examples.
The reader can get more information on the weight distribution of irreducible
cyclic codes in~\cite{DY}.

For reducible cyclic codes, the determination of weight distribution involves evaluation
of exponential sums. Despite of some cases where a neat expression is available (see \cite{DLMZ,FL,HX,LF1,LF2,LTW1,LTW2,MZLFD,M,ZLH}), the computation is generally
rather complicated. Although delicate techniques were applied to the computation,
most of these literature, to our knowledge, can only obtain the weight distribution
of reducible cyclic codes whose dual codes have two or three zeros.
The exponential sums which have been explicitly evaluated seem to
share a common feature that they attain only a few distinct
values.

In this paper, we determine the weight distribution of a class of reducible cyclic codes
whose dual codes may have arbitrarily many zeros. This goal is achieved by building a surprising
connection between the involved exponential sums and the spectrums of Hermitian forms graphs.
The rest of this paper is organized as follows. The codes we considered will be
introduced in Section \ref{sect:code}. A brief introduction to Cayley graphs and
Hermitian forms graphs is given in Section \ref{sect:CayAndHerm}. We build the connection
between exponential sums and spectrums of Hermitian forms graphs
in Section \ref{sect:weightDistri}. After presenting this connection, the weight distribution
follows immediately. A brief conclusion will be given in the last section.

\section{The code $\ca{C}_{(p,m)}$}\label{sect:code}
First we fix some notation. Let $p$ be a prime and $q=p^n$ with $n=2m$, where $m$ is odd.
Write $t=(m-1)/2$. Suppose $\pi$ is a primitive element of $\Fq$.
Let $h_0(x)$ be the minimal polynomial of $\pi^{-(p^m+1)}$ over $\Fp$. Then $\deg h_0(x)=m.$ Let $h_i(x)$ be
the minimal polynomial of $\pi^{-(p^{2i-1}+1)}$ over $\Fp$,
where $1\le i \le t$.
For any integer $l>1$ with $l|2m$,
we have $\pi^{-(p^{2i-1}+1)(p^{2m/l}-1)}\neq 1$, where $1\le i \le t$.
Thus $\deg h_i(x)=n$ for $1\le i \le t$.
Since for $1\le i<j\le t$, there does not exist any positive integer $k$
such that
$$
p^k(p^{2i-1}+1) \equiv p^{2j-1}+1 \pmod{q-1},
$$
the elements $\pi^{-(p^{2i-1}+1)}$ and $\pi^{-(p^{2j-1}+1)}$ have distinct minimal polynomials
over $\Fp$. So the polynomials $h_i(x)$ are distinct for $0 \le i \le t$.

Let $\ca{C}_{(p,m)}$ be the cyclic code with parity check polynomial $h_0(x)h_1(x)\cdots h_{t}(x)$
over $\Fp$.
Then the code $\ca{C}_{(p,m)}$ is the dual of a cyclic code with $t+1$ zeros
and $\textup{dim}{}_{\Fp} \ca{C}_{(p,m)}=m^2$.
Let $\Tr_{i}^{j}$ denote the trace mapping from $\mathbb{F}_{p^{j}}$ to $\mathbb{F}_{p^{i}}$.
The codewords in $\ca{C}_{(p,m)}$ can be expressed as
$$
\Bc_{[\alpha_0, \alpha_1,\ldots,\alpha_{t}]}=(c_0,c_1,\ldots,c_{q-2})\quad
(\alpha_0 \in \Fpm, \alpha_1,\ldots,\alpha_{t} \in \Fq)
$$
where
$$
c_i=\Tr_{1}^{m}(\alpha_0\pi^{i(p^{m}+1)})+\sum_{j=1}^{t}\Tr_{1}^{n}(\alpha_j\pi^{i(p^{2j-1}+1)}),
$$
for $0 \le i \le q-2$ (see \cite{Del}). Hence the Hamming weight of the codeword $\Bc_{[\alpha_0, \alpha_1,\ldots,\alpha_{t}]}$ is
\begin{equation*}
w_H(\Bc)=p^{n-1}(p-1)-\frac{1}{p}\sum_{a\in \Fp^*}T(a\alpha_0,a\alpha_1,\ldots,a\alpha_{t}),
\end{equation*}
where
$$
T(\alpha_0, \alpha_1,\ldots,\alpha_{t})=\sum_{x \in \Fq}\zeta_p^{\Tr_{1}^{m}(\alpha_0x^{p^{m}+1})+\sum_{j=1}^{t}\Tr_{1}^{n}(\alpha_jx^{p^{2j-1}+1})}.
$$

Generally speaking, it is very difficult to obtain the value distribution of
$T(\alpha_0, \alpha_1,\ldots,\alpha_{t})$ for $\alpha_0 \in \Fpm$, $\alpha_1,\ldots,\alpha_{t} \in \Fq$,
especially when $t$ is large.
In the next section, we will build a surprising connection between the multiset
$\{T(\alpha_0, \alpha_1,\ldots,\alpha_{t})\mid \alpha_0 \in \Fpm$, $\alpha_1,\ldots,\alpha_{t} \in \Fq\}$
and the eigenvalues of Hermitian forms graphs, which can reduce the complexity of
computation for $T(\alpha_0, \alpha_1,\ldots,\alpha_{t})$ remarkably.

\section{Cayley Graphs and Hermitian Forms Graphs}\label{sect:CayAndHerm}
Now we record some known results on Cayley graphs and Hermitian forms graphs.
\subsection{Cayley graphs}
Let $G$ be a finite group and $D\subset G$ be a subset.
The {\it Cayley graph} $Cay(G,D)$
on $G$ with connection set $D$ is the directed graph with vertex set $G$
and edge set $\{(g,h)\mid g,h \in G ,\; hg^{-1} \in D\}$.

Define $D^{(-1)}=\{d^{-1}\mid d \in D\}$. Then $Cay(G,D)$ is undirected if $D=D^{(-1)}$.
Furthermore, $Cay(G,D)$ is $k$-regular with $k=|D|$. If $G$ is a finite abelian group,
it is easy to compute the spectrum of $Cay(G,D)$.
For any character $\chi$ of $G$, define $\chi(D)=\sum_{d\in D}\chi(d)$. The character
group of $G$ is denoted by $\widehat{G}$, with $|\widehat{G}|=|G|$.

\begin{lemma}\label{lemma:spectrumOfCayleyGraph}
Let $\Gamma=Cay(G,D)$ be a Cayley graph on a finite abelian group $G$ with connection set $D$. Suppose $A=A(\Gamma)$ is the adjacency matrix of $\Gamma$. Then each character $\chi$ of $G$ corresponds to an eigenvector of $A$ with eigenvalue $\chi(D)$. In particular, the spectrum of $\Gamma$ is the multiset $\{\chi(D)\mid \chi \in \widehat{G}\}$.
\end{lemma}
\begin{proof}
Let $\chi$ be a character of $G$. Let $e_{\chi}$ be the column vector $(\chi(g))_{g \in G}$. For any $h\in G$, we have
$$
(Ae_{\chi})_h=\sum_{g\thicksim h}\chi(g)=\left(\sum_{d\in D}\chi(d)\right)\chi(h)=\chi(D)\chi(h).
$$
Hence, $e_{\chi}$ is an eigenvector of $A$ with eigenvalue $\chi(D)$. All characters in $\widehat{G}$ give rise to $|G|$ linearly independent eigenvectors, thus one obtains the spectrum of Cayley graph $\Gamma$ via the character group $\widehat {G}$.
\end{proof}

\subsection{Hermitian forms graphs}
Let $V=\Frtwo^d$, where $r$ is a prime power. For any $x \in \Frtwo$,
its conjugate $\overline{x}$ is defined by $\overline{x}=x^r$.
A matrix $H$ over $\Frtwo$ is called Hermitian if $H=H^{{}^*}$,
where $H^{{}^*}$ is the conjugate transpose of $H$.
Let $\cH$ denote the abelian group formed by all $d\times d$ Hermitian matrices over $\Frtwo$
under the matrix addition. Clearly, we have $|\cH|=r^{d^2}$.
The {\it Hermitian forms graph} on $V$ is the graph whose vertices are the
elements of $\cH$ and in which $H_1,H_2\in \cH$ are adjacent whenever $rank(H_1-H_2)=1$.
Equivalently, the Hermitian forms graph is the Cayley graph $Cay(\cH,\cD)$,
where $\cD=\{H\in\cH\mid rank(H)=1\}$.
A $d\times d$ Hermitian matrix $H$ of rank $1$
can be written as $H=\Ba^T\overline{\Ba}$, where $\Ba=(a_1,\ldots,a_d)\in V$
and $\overline{\Ba}=(\overline{a_1},\overline{a_2},\ldots,\overline{a_d})$.
Since for any $\Ba,\Bb \in V$, $\Ba^T\overline{\Ba}=\Bb^T\overline{\Bb}$ if and only if $\Ba=\gamma\Bb$ for some $(r+1)$-th root of unity $\gamma$, we have
$|\cD|=(r^{2d}-1)/(r+1)$.

It is well known that the Hermitian forms graph
on $V$ is a distance regular graph with classical parameters
$(d,b,\alpha,\beta)=(d,-r,-r-1,-(-r)^d-1)$ \cite[Table 6.1]{BCN}.
The eigenvalues of the Hermitian forms graph were first computed by Stanton \cite{Stan}.
Here we quote the more accessible formulas given in \cite{BCN}.
For any integers $j\ge i\ge 0$ and $b\neq 0,1$, the Gaussian binomial coefficients with basis $b$ are defined by
\begin{eqnarray*}
\bigg[ {j \atop i}\bigg]_b=\left\{
                            \begin{array}{ll}
                             \prod_{l=0}^{i-1}\frac{b^j-b^l}{b^i-b^l} & \text{if $i\ge 1$},\\
                             1                                        & \text{if $i = 0$}.
                            \end{array}
                            \right.
\end{eqnarray*}

\begin{lemma}\cite[Corollary 8.4.4]{BCN}
\label{lemma:eigenOfHermGraph}
Let $V=\Frtwo^d$, where $r$ is a prime power. The Hermitian forms graph
defined on $V$ has eigenvalues
$$
\theta_0=\frac{r^{2d}-1}{r+1}, \quad \theta_j=\frac{r^{2d}-1}{r+1}+(-r)^{2d-j}\bigg[ {j \atop 1}\bigg]_{(-r)},
$$
for $1 \le j \le d$. Their corresponding multiplicities are
$$
f_0=1, \quad f_j=\bigg[ {d \atop j}\bigg]_{(-r)}\prod_{l=0}^{j-1}[(-1)^{d+1}r^d+(-1)^{l+1}r^l],
$$
where $1\le j\le d$.
\end{lemma}

\section{The weight distribution of the code $\ca{C}_{(p,m)}$}\label{sect:weightDistri}
Throughout this section, $p,q,n,m,t$ are defined as in Section \ref{sect:code}.
Consider the abelian group
$$
\cG=\Fpm \times \underbrace{\Fq \times \Fq \times \cdots \times \Fq}_{t},
$$
and its subset
$$
\cS=\big\{(x^{p^m+1},x^{p+1},x^{p^{3}+1},\ldots,x^{p^{m-2}+1})\mid x\in \Fq^*\big\}.
$$
It is easy to see that $|\cS|=(q-1)/(p+1)$.
Let $W=\Fptwo^m$ and $\cH$ be the abelian group consisting of all
$m \times m$ Hermitian matrices over $\Fptwo$.
Let $\cD=\{H \in \cH \mid rank(H)=1\}.$ Clearly, the Hermitian forms graph
on $W$ is the Cayley graph $Cay(\cH,\cD)$.
The following lemma shows that the Cayley graph $Cay(\cG,\cS)$
shares the same spectrum with $Cay(\cH,\cD)$.

\begin{lemma}\label{lemma:isomorphism}
For odd $m$, the Hermitian forms graph $\Gamma_1$ on $W=\Fptwo^m$ is isomorphic to
the Cayley graph $\Gamma_2=Cay(\cG,\cS)$. In particular, $\Gamma_1$ and $\Gamma_2$ have the same spectrum.
\end{lemma}
\begin{proof}
Since $\Gamma_1$ is just the Cayley graph $Cay(\cH,\cD)$, the result will
follow immediately if we can find
a group isomorphism $\varphi$ from $\cH$ to $\cG$ satisfying that $\varphi(\cD)=\cS$.

Let $e_1,\ldots,e_m$ be a basis of $\Fq$ over $\Fptwo$ and $\Be=(e_1,e_2,\ldots,e_m)$.
For any $H\in\cH$ and $\Bx,\By \in \Fq^m$, we define $f_H(\Bx,\By)=\Bx H\By^T$, where $\By^T$ is the transpose of $\By$.
Now we construct a mapping $\varphi$ from $\cH$ to $\cG$ by sending $H\in\cH$ to
$$
\varphi(H)=(f_H(\Be,\Be^{p^m}),f_H(\Be,\Be^{p}),f_H(\Be,\Be^{p^3}),\ldots,f_H(\Be,\Be^{p^{m-2}})),
$$
where $\Be^{s}:=(e_1^s,e_2^s,\ldots,e_m^s)$ for any integer $s$. It is straightforward to verify that $\varphi(H) \in \cG$ and $\varphi$ is a group homomorphism.

Now we want to show that $\varphi$ is an isomorphism.
First, we prove that $\varphi$ is injective. For a matrix $H=(h_{ij})$ and an integer $s$,
we denote $H{}^{s}=(h_{ij}^s)$.
Suppose $\varphi(H)=(0,0,\ldots,0)$, i.e., $\Be H(\Be^{p^m})^T=f_H(\Be,\Be^{p^m}) =0$ and
$\Be H(\Be^{p^{2i-1}})^T=f_H(\Be,\Be^{p^{2i-1}})=0$ for $1\le i\le t$.
By rising all entries of $\Be H(\Be^{p^{2i-1}})^T$ to their $p^{2m-2i+1}$-th power, we have
$$
\Be^{p^{2m-2i+1}} H^{{}^{p^{2m-2i+1}}}\Be^T=\Be^{p^{2m-2i+1}} H^{{}^p}\Be^T=0,
$$
which gives
$$
\Be H(\Be^{p^{2m-2i+1}})^T=0.
$$
So we obtain $\Be H\Psi=(0,0,\ldots,0)$, where
$$
\Psi=(({\Be^p})^T,({\Be^{p^3}})^T,\cdots,({\Be^{p^m}})^T,\cdots,(\Be^{p^{2m-3}})^T,(\Be^{p^{2m-1}})^T).
$$
By the choice of $\Be$, it can be shown that $\Psi$ is a nonsingular matrix \cite[Corollary 2.38]{LN}.
Therefore $\Be H=(0,0,\ldots,0)$, which implies that $H$ is a zero matrix. Consequently, $\varphi$ is injective. On the other hand, a direct calculation shows that $|\cH|=|\cG|=p^{m^2}$. Hence $\varphi$ is an isomorphism.

For any $H \in \cD$, we have $H=\Ba^{T}\Ba^p$ for some $\Ba=(a_1,a_2,\ldots,a_m)\in W$, where $\Ba^p=(a_1^p,a_2^p,\ldots,a_m^p)$.
Therefore,
\begin{eqnarray*}
\varphi(H)&=&(\Be H(\Be^{p^m})^T,\Be H(\Be^{p})^T,\ldots,\Be H(\Be^{p^{m-2}})^T)\\
            &=&(\Be\Ba^{T}\Ba^p(\Be^{p^m})^T,\Be\Ba^{T}\Ba^p(\Be^{p})^T,\ldots,
                \Be\Ba^{T}\Ba^p(\Be^{p^{m-2}})^T)\\
            &=&(x^{p^m+1}, x^{p+1},\ldots, x^{p^{m-2}+1}),
\end{eqnarray*}
where $x=\Be \Ba^{T} \in \Fq^{*}$. Thus we obtain $\varphi(\cD)\subset \cS$. Since
$$|\cD|=(p^{2m}-1)/(p+1)=(q-1)/(p+1)=|\cS|,$$
we have $\varphi(\cD)=\cS$. So we have proved that $\varphi$ is an isomorphism from $\cH$ to $\cG$
sending the connection set $\cD$ to $\cS$. Therefore $\Gamma_1$ is isomorphic to $\Gamma_2$,
and they have the same spectrum.
\end{proof}

From Lemma \ref{lemma:isomorphism} and Lemma~\ref{lemma:eigenOfHermGraph},
the eigenvalues of $\Gamma_2$ and their multiplicities are known.
On the other hand, the eigenvalues of $\Gamma_2$ can be expressed using
Lemma~\ref{lemma:spectrumOfCayleyGraph}.

Note that
$$
\widehat{\cG}=\{\chi_{(\alpha_0,\alpha_1,\ldots,\alpha_t)}\mid\alpha_0\in\Fpm,\alpha_1,\ldots,\alpha_{t}\in\Fq\},
$$
where
$$
\chi_{(\alpha_0,\alpha_1,\ldots,\alpha_t)}(u)=\zeta_p^{\Tr_{1}^{m}(\alpha_0u_0)+\sum_{j=1}^{t}\Tr_{1}^{n}(\alpha_ju_j)},
$$
for any $u=(u_0,u_1,\ldots,u_t) \in \cG$. By Lemma~\ref{lemma:spectrumOfCayleyGraph}, the eigenvalues of $\Gamma_2$ are
\begin{eqnarray*}
\lefteqn{\chi_{(\alpha_0,\alpha_1,\ldots,\alpha_t)}(\cS)}\\\\
&=&\sum_{u\in \cS}\zeta_p^{\Tr_{1}^{m}(\alpha_0u_0)+\sum_{j=1}^{t}\Tr_{1}^{n}(\alpha_ju_j)}\\
&=&\frac{1}{p+1}\sum_{x\in\Fq^*}\zeta_p^{\Tr_{1}^{m}(\alpha_0x^{p^{m}+1})+\sum_{j=1}^{t}\Tr_{1}^{n}(\alpha_jx^{p^{2j-1}+1})}\\
&=&\frac{1}{p+1}\big(T(\alpha_0,\alpha_1,\ldots,\alpha_{t})-1\big),
\end{eqnarray*}
for all $\alpha_0 \in \Fpm, \alpha_1,\ldots,\alpha_{t} \in \Fq.$
Therefore we have
$$
T(\alpha_0,\alpha_1,\ldots,\alpha_{t})=(p+1)\chi_{(\alpha_0,\alpha_1,\ldots,\alpha_t)}(\cS)+1.
$$
By Lemma~\ref{lemma:eigenOfHermGraph}, the eigenvalues of $\Gamma_2$ are all rational numbers.
Thus, we have $T(\alpha_0,\alpha_1,\ldots,\alpha_{t})\in \Q$.
For any $a\in \Fp^*$, there exists an automorphism $\sigma_a \in \Gal(\Q(\zeta_p)/\Q)$
with $\sigma_a(\zeta_p)=\zeta_p^a$. Hence
\begin{eqnarray*}
\sum_{a\in\Fp^*}T(a\alpha_0,a\alpha_1,\ldots,a\alpha_{t})
&=&\sum_{a\in\Fp^*}\sigma_a\big(T(\alpha_0,\alpha_1,\ldots,\alpha_{t})\big)\\
&=&(p-1)T(\alpha_0,\alpha_1,\ldots,\alpha_{t}).
\end{eqnarray*}
Consequently, the Hamming weight of $\Bc_{[\alpha_0, \alpha_1,\ldots,\alpha_{t}]}$ is
\begin{equation}\label{equ:weight2}
\begin{split}
w_H(\Bc)&=p^{n-1}(p-1)-\frac{1}{p}\sum_{a\in \Fq^*}T(a\alpha_0,a\alpha_1,\ldots,a\alpha_{t})\\
      &=p^{n-1}(p-1)-\frac{p-1}{p}\,T(\alpha_0,\alpha_1,\ldots,\alpha_{t})\\
      &=p^{n-1}(p-1)-\frac{p-1}{p}\big(1+(p+1)\chi_{(\alpha_0,\alpha_1,\ldots,\alpha_t)}(\cS)\big).
\end{split}
\end{equation}
Now the weight distribution of the code $\ca{C}_{(p,m)}$ follows directly from
Equation~(\ref{equ:weight2}) and  Lemma~\ref{lemma:eigenOfHermGraph}. In the following theorem, we use $[l,k,d]$ code as the notation for a $k$-dimensional linear code of length $l$ with minimum distance $d$.
\begin{theorem}\label{thm:weightDist}
For any odd integer $m$, the weight distribution of  the code  $\ca{C}_{(p,m)}$ is as follows:
$$
A_i=\left\{\begin{array}{ll}
    1   & \textup{ if } i=0,\\
    f_j & \textup{ if } i=w_j,\\
    0   & \textup{otherwise},
    \end{array}
    \right.
$$
where
$$
w_j=(p^{2m}-p^{2m-1})\left(1-\frac{1}{(-p)^j}\right),
$$
and
$$
f_j=\bigg[ {m \atop j}\bigg]_{(-p)}\prod_{l=0}^{j-1}\big(p^m-(-p)^l\big)
$$
for $1\le j\le m$. In particular, the code $\ca{C}_{(p,m)}$ is a
$[p^{2m}-1,m^2,(p^{2m}-p^{2m-1})(1-p^{-2})]$ cyclic code.
\end{theorem}

For a code $\cC$ with weight distribution $\{A_0,A_1,\ldots,A_l\}$, define its {\it weight enumerator} as
$$
\sum_{i=0}^{l} A_i x^i.
$$
The weight enumerator provides a succinct way to express the weight distribution. For the purpose of illustration, we give two examples below.

\begin{example}
Let $p=3,$ and $m=3.$ The code $\ca{C}_{(3,3)}$ is then a $[728,9,432]$
code over GF$(3)$ with the weight enumerator
$$
1+5460x^{432}+14040x^{504}+182x^{648}.
$$
\end{example}

\begin{example}
Let $p=2,$ and $m=5.$ The code $\ca{C}_{(2,5)}$ is then a $[1023,25,384]$
code over GF$(2)$ with the weight enumerator
\begin{align*}
&1+57970x^{384}+12985280x^{480}+18887680x^{528}\\
&+1623160x^{576}+341x^{768}.
\end{align*}
\end{example}

\section{Conclusion}\label{sect:conclusion}

In the study of cyclic codes, researchers have established the connections between their weight distribution and other mathematical objects, such as Gauss sums (see~\cite{FM,Moi}), algebraic curves (see~\cite{Schoof,Vl,WTQYX}), as well as quadratic forms (see~\cite{FL,LF1,LF2}). In this paper, we found an elegant connection between the weight distribution of a class of cyclic codes and the spectrums of certain distance regular graphs. In this way, the weight distribution of these codes follows from the known spectrums of Hermitian forms graphs. The dual codes of this family of cyclic codes may have arbitrarily many zeros, while most previously known results are obtained in the case where the dual codes have no more than three zeros.

\section*{Acknowledgement}

The authors express their gratitude to the two anonymous
reviewers for their detailed and constructive comments which
are very helpful to the improvement of this paper, and to Dr. Navin Kashyap,
the Associate Editor, for his excellent editorial job.

\bibliographystyle{IEEEtranS}
\bibliography{Reference}
\end{document}